\documentclass[journal, onecolumn]{IEEEtran}
\usepackage{amsmath} % assumes amsmath package installed
\usepackage{amssymb}  % assumes amsmath package installed
\usepackage{amsfonts}  % assumes amsmath package installed
\usepackage{amsthm}  % assumes amsmath package installed

\usepackage{pdflscape}
\usepackage{rotating}
\usepackage{color}
\usepackage{epstopdf}

\usepackage[labelfont=bf]{caption}
\usepackage{subcaption}
\usepackage{placeins}

\usepackage{cite}

\newcommand{\col}{\mbox{col}}

\def\calg{{\cal G}}

\def\calo{{\cal O}}

\def\bfy{{\bf y}}
\def\bfm{{\bf m}}
\def\bfM{{\bf M}}
\def\bfY{{\bf Y}}

\def\L2e{{\cal L}_{2e}}

\def\rea{\mathbb{R}}

\def\diag{\mbox{diag}}
\def\adj{\mbox{adj}}

\def\begequarr{\begin{eqnarray}}
\def\endequarr{\end{eqnarray}}
\def\begequarrs{\begin{eqnarray*}}
\def\endequarrs{\end{eqnarray*}}
\def\begarr{\begin{array}}
\def\endarr{\end{array}}
\def\begequ{\begin{equation}}
\def\endequ{\end{equation}}
\def\lab{\label}
\def\begdes{\begin{description}}
\def\enddes{\end{description}}
\def\begenu{\begin{enumerate}}
\def\begite{\begin{itemize}}
\def\endite{\end{itemize}}
\def\endenu{\end{enumerate}}

\def\lef[{\left[\begin{array}}
\def\rig]{\end{array}\right]}
\def\qed{\hfill$\Box \Box \Box$}
\def\begcen{\begin{center}}
\def\endcen{\end{center}}
\def\begrem{\begin{remark}\rm}
\def\endrem{\end{remark}}

\def\call{{\cal L}}

\def\et{\epsilon_t}

\def\diag{\mbox{diag}}

\def\rea{\mathbb{R}}

\newtheorem{assumption}{Assumption}
\newtheorem{proposition}{Proposition}
\newtheorem{lemma}{Lemma}
\newtheorem{remark}{Remark}

\begin{document}

\title{Performance Enhancement of Parameter Estimators via Dynamic Regressor Extension and Mixing}
\author{Stanislav~Aranovskiy$^{1,2}$,~\IEEEmembership{Member,~IEEE,}
        Alexey~Bobtsov$^{2}$,~\IEEEmembership{Senior Member,~IEEE,}
				Romeo~Ortega$^{3}$,~\IEEEmembership{Fellow Member,~IEEE,}
        Anton~Pyrkin$^{2}$,~\IEEEmembership{Member,~IEEE}% <-this % stops a space
\thanks{$^{1}$Stanislav Aranovskiy is with the School of Automation, Hangzhou Dianzi University, Xiasha Higher Education Zone, Hangzhou, Zhejiang 310018, P.R.China}		
\thanks{$^{2}$Stanislav Aranovskiy, Alexey Bobtsov and Anton Pyrkin are with the Department of Control Systems and Informatics, ITMO University, Kronverkskiy av. 49, Saint Petersburg, 197101, Russia :  {\tt\small aranovskiysv@niuitmo.ru}}
\thanks{$^{3}$Romeo Ortega is with the LSS-Supelec, 3, Rue Joliot-Curie, 91192 Gif--sur--Yvette, France : {\tt\small ortega@lss.supelec.fr}}
}

\maketitle
%
%%%%%%%%%%%%%%%%%%%%%%%%%%%%%%%%%%%%%%%%%%%%%%%%%%%%%%%%%%%%%
\begin{abstract}
A new way to design parameter estimators with enhanced performance is proposed in the paper. The procedure consists of two stages, first, the generation of new regression forms via the application of a dynamic operator to the original regression. Second, a suitable mix of these new regressors to obtain the final desired regression form. For classical linear regression forms the procedure yields a new parameter estimator whose convergence is established without the usual  requirement of regressor persistency of excitation.  The technique is also applied to nonlinear regressions with ``partially" monotonic  parameter dependence---giving rise again to estimators with enhanced performance. Simulation results illustrate the  advantages of the proposed procedure in both scenarios. 
\end{abstract}
%%%%%%%%%%%%%%%%%%%%%%%%%%%%%%%%%
%%%%%%%%%%%%%%%%%%%%%%%%%%%
\section{Introduction}
\lab{sec1}
%%%%%%%%%%%%%%%%%%%%%%%%%%%%
% 
A new procedure to design parameter identification schemes is proposed in this article. The procedure, called Dynamic Regressor Extension and Mixing (DREM), consists of two stages, first, the generation of new regression forms via the application of a dynamic operator to the data of the original regression. Second, a suitable mix of these new data to obtain the final desired regression form to which standard parameter estimation techniques are applied.

The DREM procedure is applied in two different scenarios. First, for linear regression systems, it is used to generate a new parameter estimator whose convergence is ensured {\em without a persistency of excitation} (PE) condition on the regressor. It is well known that standard parameter estimation algorithms applied to linear regressions give rise to a linear time--varying system, which is exponentially stable if and only if a certain PE condition is imposed---this fundamental result constitutes one of the main building blocks of identification and adaptive control theories \cite{LJU,SASBOD}. To the best of the authors' knowledge there is no systematic way to conclude asymptotic stability for this system without this assumption, which is rarely verified in applications. Relaxation of the PE condition is a challenging theoretical problem and many research works have been devoted to it in various scenarios, see {\em e.g.}, \cite{ARAetal,CHOJOH,EFIFRA,JAImsc,MISDAR,PANYU, ALCOSP} and references therein. Due to its practical importance research on this topic is of great current interest. 

The second parameter estimation problem studied in this article is when the parameters enter {\em nonlinearly} in the regression form. It is well known that nonlinear parameterizations are inevitable in any realistic practical problem. On the other hand, designing parameter identification algorithms for nonlinearly parameterized regressions is a difficult poorly understood problem. An interesting case that has recently been explored in the literature is when the dependence with respect to the parameters  exhibit some {\em monotonicity} properties; see \cite{LIUTAC,LIUSCL,TYUetal}. Unfortunately, it is often the case that this property holds true {\em only for some} of the functions entering in the regression stymying the application of the proposed techniques. Our second contribution is the use of the DREM technique to ``isolate" the good nonlinearities and be able to exploit the monotonicity to achieve consistent parameter estimation for nonlinearly parameterised regressions with {factorisable} nonlinearities---{\em not imposing PE conditions}.

The remaining of the paper is organized as follows. The  DREM technique is first explained with its application to linear regressions in Section \ref{sec2}. In Section \ref{sec3} DREM is used for nonlinear factorisable regressions with ``partially" monotonic  parameter dependence. In both sections representative simulation examples are presented. Some concluding remarks and future research are given in Section
\ref{sec4}.\\

%%%%%%%%%
\noindent {\bf Notation} For $x \in \rea^n$,  $|x|$ is the Euclidean norm. All functions and mappings in the paper are assumed sufficiently smooth. For functions of scalar argument $g:\rea \to \rea^s$, $g'$ denotes its first  order derivative. For functions $V: \rea^n \to \rea$ we define the operator $\nabla V :=
(\frac{\partial V}{\partial x})^\top$. Also, for mappings $\Phi:\rea^n \times \rea^q \to \rea^n$ we define its (transposed) Jacobian matrix  $\nabla_x \Phi(x,\theta) := [\nabla \Phi_1(x,\theta),\dots, \nabla \Phi_n(x,\theta)]$. For the distinguished element $x_\star \in \rea^n$ and any mapping $F:\rea^n \to \rea^s$ we denote $F_\star :=F(x_\star )$.  
%
%
%%%%%%%%%%%%%%%%%%%%%%%%%%%%
\section{Consistent Estimation for Linear Regressions without PE}
\lab{sec2}
%%%%%%%%%%%%%%%%%%%%%%%%%%%%
%
In this section the DREM technique is applied to classical linear regressions. The main contribution is the removal of the---often overly restrictive---assumption of regressor PE  to ensure parameter convergence.
\subsection{Standard procedure and the PE condition}
\lab{subsec21}
%%%%%%%%%%%%%%%%%%%%%%%%%%%%
%
Consider the basic problem of on--line estimation of the constant parameters of the $q$--dimensional linear regression
\begequ
\lab{y}
y(t)=m^\top(t)\theta,
\endequ
where\footnote{When clear from the context, in the sequel the arguments of the functions are omitted.} $y:\rea_+ \to \rea$ and $m: \rea_+ \to \rea^q$ are known, bounded functions of time and $\theta \in \rea^q$ is the vector of unknown parameters. The standard gradient estimator
\begequ
\lab{parest0}
\dot {\hat \theta}=\Gamma m (y - m^\top \hat \theta),
\endequ
with a positive definite adaptation gain $\Gamma \in \rea^{q \times q}$ yields the error equation
\begequ
\lab{ltvsys}
\dot {\tilde \theta}=-\Gamma m(t) m^\top(t) \tilde \theta,
\endequ
where $\tilde{\theta} := \hat{\theta}-\theta$ are the parameter estimation errors. It is well--known \cite{LJU,SASBOD} that the zero equilibrium of the linear time--varying system \eqref{ltvsys} is (uniformly) exponentially stable if and only if the regressor vector $m$ is PE, that is, if
\begequ
\lab{pe}
\int_t^{t+T} m(s) m^\top(s)ds \geq \delta I_q,
\endequ
for some $T,\delta >0$ and for all $t \geq 0$, which will be denoted as $m(t) \in \mbox{PE}$. If $m(t) \notin \mbox{PE}$, which happens in many practical circumstances, very little can be said about the asymptotic stability of \eqref{ltvsys}, hence about the convergence of the parameter errors to zero.

\begrem
In spite of some erroneous claims \cite{YUART}, it is well known that the PE conditions for the gradient estimator presented above and more general estimators---like (weighted) least squares---exactly coincide \cite{DASHUA}. Since the interest in the paper is to relax the PE condition, attention is restricted to the simple gradient estimator. 
\endrem

\begrem
To simplify the notation it has been assumed above that the measurement signal $y$ is one--dimensional. As will become clear below DREM is applicable also for the vector case.
\endrem 
\subsection{Proposed dynamic regressor extension and mixing procedure}
\lab{subsec22}
%%%%%%%%%%%%%%%%%%%%%%%%%%%%
%
To overcome the limitation imposed by the PE condition the DREM procedure generates $q$ new, one--dimensional, regression models to  independently estimate each of the parameters under conditions on the regressor $m$ that differ from the PE condition \eqref{pe}.

The first step in DREM is to introduce $q-1$ {\em linear, $\call_\infty$--stable} operators $H_i: \call_\infty \to \call_\infty,\;i \in \{1,2,\dots,q-1\}$, whose output, for any bounded input, may be decomposed as
\begequ
\lab{defh}
(\cdot)_{f_i}(t):=[H_i(\cdot)](t) + \epsilon_t,
\endequ
with $\epsilon_t$ is a (generic) exponentially decaying term. For instance, the operators $H_i$ may be simple, exponentially stable {\em LTI filters} of the form
\begequ
\lab{ltifil}
H_i(p)=\frac{\alpha_i}{p + \beta_i},
\endequ
with $p:=\frac{d}{dt}$ and $\alpha_i\ne 0$, $\beta_i>0$; in this case $\et$ accounts for the effect of the initial conditions of the filters. Another option of interest are {\em delay operators}, that is
$$
[H_i(\cdot)](t):=(\cdot)(t-d_i),
$$
where $d_i \in \rea_+$.

Now, we apply these operators to the regressor equation \eqref{y} to get the filtered regression\footnote{To simplify the presentation in the sequel we will neglect the $\epsilon_t$ terms, which will be incorporated in the analysis later.}
\begequ
\lab{yf}
y_{f_i}= m^\top_{f_i}\theta.
\endequ
Piling up the original regressor equation \eqref{y} with the $q-1$ filtered regressors we can construct the extended regressor system
\begequ
\label{YM}
Y_e(t) = M_e(t) \theta,
\endequ
where we defined $Y_e:\rea_+ \to \rea^q$ and $M_e:\rea_+ \to \rea^{q \times q}$ as
\begequ
\lab{yama}
Y_e:=	\begin{bmatrix} y \\ y_{f_1} \\ \vdots \\ y_{f_{q-1}} \end{bmatrix},\;M_e:=\begin{bmatrix} m^\top \\ m^\top_{f_1} \\ \vdots \\ m^\top_{f_{q-1}} \end{bmatrix}.
\endequ
{Note that, because of the $\call_\infty$--stability assumption of $H_i$, $Y_e$ and $M_e$ are bounded.} Premultiplying  \eqref{YM} by the {\em adjunct matrix} of $M_e$ we get $q$ scalar regressors of the form 
\begequ
\lab{scareg}
Y_i(t) = \phi(t) \theta_i
\endequ
with $i \in \bar q:= \{1,2,\dots,q\}$, where we defined the determinant of $M_e$ as
\begequ
\lab{phi}
\phi(t):=\det \{M_e(t)\}.
\endequ
and the vector $Y:\rea_+ \to \rea^q$
\begequ
\lab{Y}
Y(t) := \adj\{M_e(t)\} Y_e(t).
\endequ

The estimation of the parameters $\theta_i$ from the scalar regression form \eqref{scareg} can be easily carried out via 
\begequ
\lab{decest}
\dot{\hat{\theta}}_i = \gamma_i\phi (Y_i - \phi \hat\theta_i) ,\;i \in \bar q,
\endequ
with adaptation gains $\gamma_i>0$. From \eqref{scareg} it is clear that the latter equations are equivalent to
\begequ
\lab{errequ}
\dot{\tilde{\theta}}_i = -\gamma_i \phi^2 \tilde\theta_i,\;i \in \bar q.
\endequ
Solving this simple scalar differential equation we conclude that
\begequ
\lab{equsta}
 \phi(t) \notin \call_2\quad \Longrightarrow \quad \lim_{t\to \infty} \tilde \theta_i(t)=0,
\endequ
with the converse implication also being true.

The derivations  above establish the following proposition.

\begin{proposition}\em
Consider the $q$--dimensional linear regression \eqref{y} where $y:\rea_+ \to \rea$ and $m: \rea_+ \to \rea^q$ are known, bounded functions of time and $\theta \in \rea^q$ is the vector of unknown parameters. Introduce  $q-1$ linear, $\call_\infty$--stable operators $H_i: \call_\infty \to \call_\infty,\;i \in \{1,2,\dots,q-1\}$ verifying \eqref{defh}. Define the vector $Y_e$ and the matrix $M_e$ as given in \eqref{yama}. Consider the estimator \eqref{decest} with $\phi$ and $Y_i$ defined in \eqref{phi} and \eqref{Y}, respectively. The implication \eqref{equsta} holds. 

\qed
\end{proposition}

\begrem
\lab{rem2}
It is important to underscore that for any matrix $A \in \rea^{q \times q}$
\begequ
\lab{adj}
\adj\{A\}A =\det\{A\} I_q,
\endequ 
even if $A$ is {\em not full rank}, \cite{Lancaster}.
\endrem

\begrem
If we take into account the presence of the exponentially decaying terms $\epsilon_t$ in the filtering operations the error equation \eqref{errequ} becomes
$$
\dot{\tilde{\theta}}_i = -\gamma_i \phi^2 \tilde\theta_i + \epsilon_t,\;i \in \bar q.
$$ 
The analysis of this equation, which establishes \eqref{equsta}, may be found in Lemma 1 of \cite{ARAetal}.
\endrem
\subsection{Discussion}
\lab{subsec23}
%%%%%%%%%%%%%%%%%%%%%%%%%%%%
%
Two natural question arise at this point.

\begite
\item[Q1.] Is the condition $\phi(t) \notin \call_2$ weaker than  $m(t) \in \mbox{PE}$? 
\item[Q2.] Given a regressor $m(t) \notin \mbox{PE}$ how to select the operators $H_i$ to enforce the condition $\phi(t) \notin \call_2$? 
\endite

Regarding these questions the following remarks are in order.

\begite
\item[R1.] It is important to recall that \cite{SASBOD}
$$
m(t) \in \mbox{PE} \quad \Longrightarrow \quad m(t) \notin \call_2.
$$
However, the converse is {\em not true} as shown by the example
$$
m(t)=\frac{1}{\sqrt{1+t}},
$$
which is neither square integrable nor PE.
\item[R2.] Consider the regressor $m(t):= [\sin(t) \ \cos(t)]^\top$ and the operator
\[
	H(p) = \frac{c(p+1)}{p^2+p+2},
\]
where $c>0$. Note that for unit frequency the operator $H$ provides zero phase shift and the magnitude gain $c$. Thus $m_{1f}(t)=c \sin(t)$, $m_{2f}(t)=c \cos(t)$ and
\[
	M_{e}(t) = \begin{bmatrix} \sin(t) & \cos(t) \\ c \sin(t) & c \cos(t) \end{bmatrix}.
\]
Obviously, $m(t) \in \mbox{PE}$, but $\det\{M_e(t)\} \equiv 0$ and $\phi(t) \in \call_2$. 
\item[R3.] From definition \eqref{pe} it is clear that the PE condition is a requirement imposed on the {\em minimal} eigenvalue of the matrix as illustrated by the equivalence 
\[
	\begin{aligned}
		\lambda_{\min}\left\{ \int_t^{t+T} m(s) m^\top(s)ds \right\} \geq {\delta} > 0 \quad \\
			\Longleftrightarrow \quad  m(t) \in \mbox{PE},
	\end{aligned}
\] 
where $\lambda_{\min}\{\cdot\}$ denotes the minimal eigenvalue. On the other hand, the condition $\phi(t) \notin \call_2$ is a restriction on {\em all} eigenvalues of the matrix $M_e$. Indeed, this is clear recalling that the determinant of a matrix is the product of all its eigenvalues and that for any two {bounded} signals $a,b:\rea_+ \to \rea$ we have
$$
a(t)b(t) \notin \call_2\quad \Longrightarrow \quad a(t) \notin \call_2 \; \mbox{and} \; b(t) \notin \call_2.
$$
Consequently, a necessary condition for parameter convergence of the estimators  \eqref{decest} is that all eigenvalues of the matrix $M_e$ are not square integrable.
\endite
Although the remarks R1--R3 do not provide a definite answer to the question Q2, they illustrate the fact that the parameter estimation procedure proposed in the paper gives rise to new convergence conditions that radically differ from the standard PE requirement. 

\subsection{An example}
\lab{subsec24}
%%%%%%%%%%%%%%%%%%%%%%%%%%%%
%
To provide some (partial) answers to the question Q2 above let us consider the simplest case of $q=2$ with $m=\col(m_1,m_2)$. In this case
\begequ
\lab{phiexa}
\phi=m_1 m_{2f} - m_{1f} m_2.
\endequ 

The proposition below identifies a class of regressors  $m(t) \notin \mbox{PE}$ but $\phi(t) \notin \call_2$ for the case of $H$ a simple LTI filter. 

\begin{proposition}\em
Define the set of differentiable functions
\[
	\begin{aligned}
		\calg:=\{g:\rea_+ \to \rea\;|\;g(t) \in \call_\infty,\;\dot g(t) \in \call_\infty,\; \dot g(t) \notin \call_2, \\
		\lim_{t\to \infty} g(t)  = \lim_{t\to \infty} \dot g(t) =  0\}
	\end{aligned}
\]
For all $g \in \calg$ the regressor
$$
m(t)=\lef[{c} 1 \\ g + \dot g \rig] \notin \mbox{PE}. 
$$
Let the operator $H$ be defined as
$$
[H(\cdot)](t) = \left[\frac{1}{p + 1}(\cdot)\right](t).
$$
The function $\phi$ defined in \eqref{phiexa} verifies $\phi(t) \notin \call_2$. 
\end{proposition}
\begin{proof}
The fact that $m(t) \notin \mbox{PE}$ is obvious because $\lim_{t\to \infty} m_2(t)  =  0$.  Now, we have that $m_{1f}=1+\et$ and from the filter equations  we get
$$
\dot m_{2f}  =  - m_{2f} + m_2.
$$
On the other hand, from the definition of $m$ we have
$$
\dot g = -g  + m_2.
$$
Substracting these two equations we get
$$
\frac{d}{dt}(m_{2f} - g)=-(m_{2f} - g),
$$
consequently $m_{2f} = g + \et$. Replacing these expressions in \eqref{phiexa} yields 
\begequarrs
\phi & = & m_{2f} - (1+\et)m_2 \\
     & = & ( g + \et) - (1+\et)(g + \dot g)\\
     & = & - \dot g +\et,
\endequarrs
where we have used the fact that $g(t) \in \call_\infty$ and $\dot g(t) \in \call_\infty$ to obtain the last equation. This completes the proof. 
\end{proof}

\begrem
\lab{rem3}
An example of a function $g \in \calg$ is
$$
g(t)=\frac{\sin t}{(1+t)^\frac{1}{2}}.
$$
The corresponding regressor is
\begin{equation} \label{eq:reg_m}
m(t)=\lef[{c} 1 \\ \frac{\sin t+\cos t}{(1+t)^\frac{1}{2}}-
\frac{\sin t}{2(1+t)^\frac{3}{2}} \rig]. 
\end{equation}
\endrem

\begrem
The choice of $\alpha=\beta=1$ for the proposed operator $H$ is made without loss of generality, as a similar proposition can be established for any exponentially stable LTI filter. The choice of the filters is, then, a degree of freedom to verify the conditions $\phi(t) \notin \call_2$.
\endrem 
\subsection{Simulation results}
\lab{subsec25}
%%%%%%%%%%%%%%%%%%%%%%%%%%%%
%
We first evaluate the performance of the classical parameters estimator \eqref{parest0} with $m(t)$ given by \eqref{eq:reg_m}. From the analysis of Subsection \ref{subsec21} we know that the LTV system \eqref{ltvsys} is stable, but it is not {\em exponentially} stable since $m(t) \not \in \mbox{PE}$, and PE is a necessary condition for exponential stability. 

The transient behavior of the parameter errors $\tilde \theta(t)$ with $\Gamma=\gamma I_2$ and $\theta=\col(-3,3)$ is shown in Fig. \ref{fig:classic_transients} for $\tilde \theta(0) = \col(3,-3)$, $\gamma=3$ and $\gamma=10$.  It is worth noting that it is not possible to conclude from the simulations whether $\tilde \theta(t)$ converges to zero asymptotically or not. The plots show that convergence has not been achieved even after  a reasonably long period of $500$. The graphs also show that increasing $\gamma$ that, in principle, should speed--up the convergence, makes the situation even worse, cf. Fig. \ref{fig:classic_transients} (a) and (b). In Fig. \ref{fig:classic_IC} we show the integral curves for $\gamma=3$ for initial conditions taken in a disk. If the adaptation gain is taken as $\Gamma=\diag\{\gamma_1,\gamma_2\}$ it is possible to improve the transient performance, but this requires a time--consuming, trial--and--error tuning stage that is always undesirable.
 
\begin{figure}[Htb]
	\centering
\subcaptionbox{$\gamma=3$}{\includegraphics[width=0.45\textwidth]{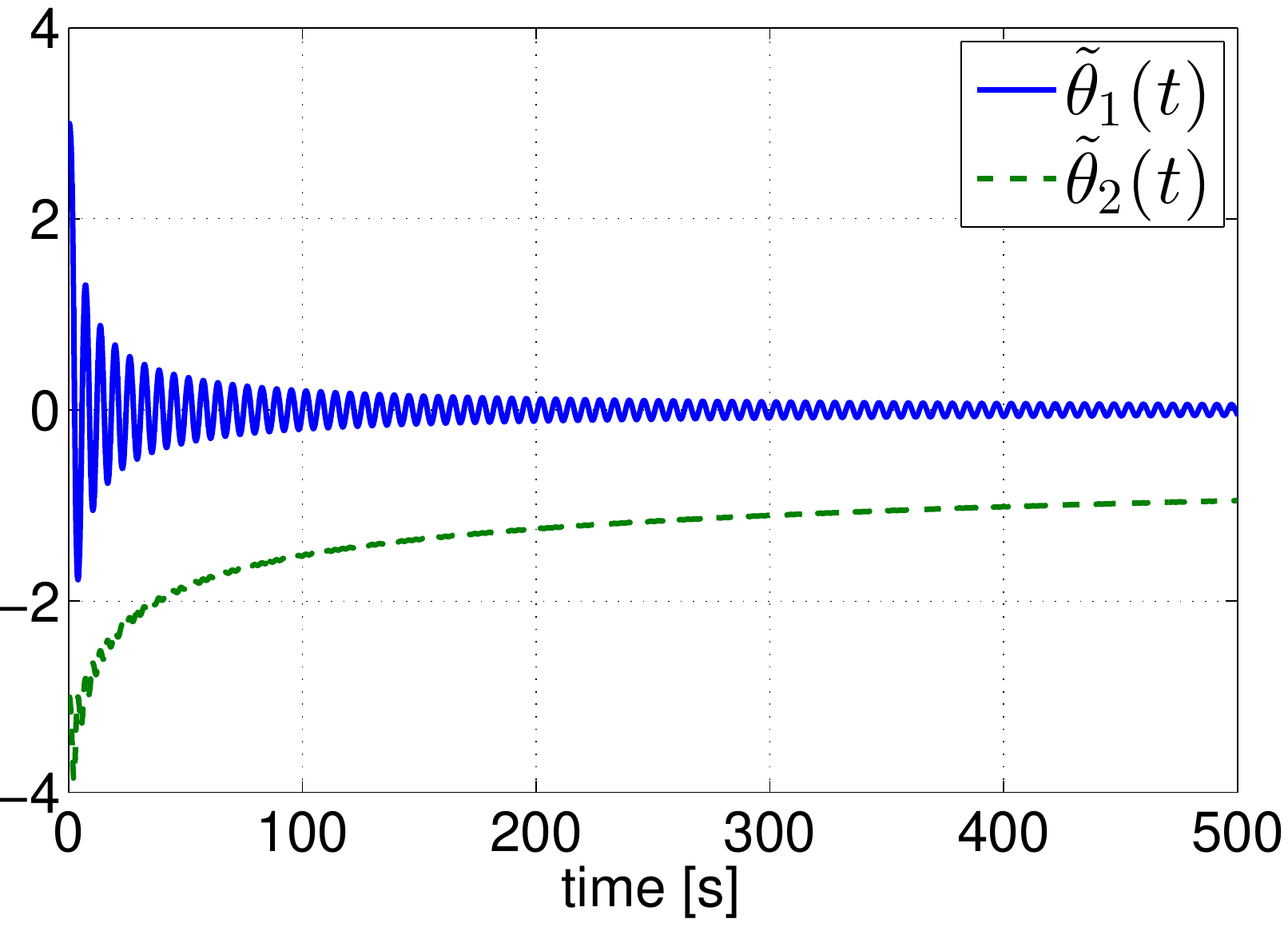}}
\subcaptionbox{$\gamma=10$}{\includegraphics[width=0.45\textwidth]{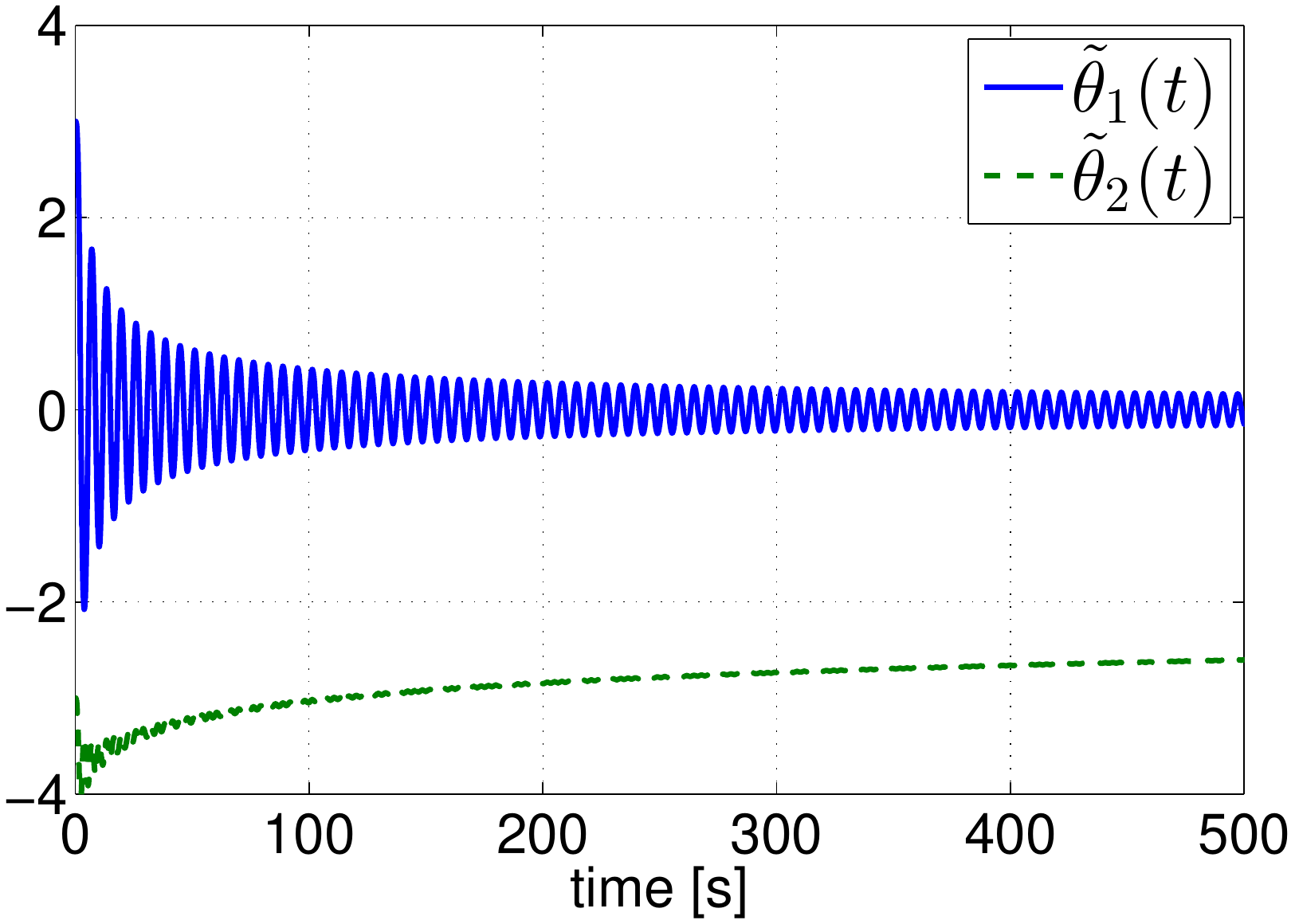}}
	\caption{Transient performance of the parameter errors $\tilde \theta(t)$ for the gradient estimator \eqref{parest0} with $m(t)$ given by \eqref{eq:reg_m}, $\tilde \theta(0) = \col(3, -3)$ and $\Gamma=\gamma I_2$, $\gamma=3,10$.}
	\label{fig:classic_transients}
\end{figure}

\begin{figure}[Htb]
	\centering
	\includegraphics[width=0.45\textwidth]{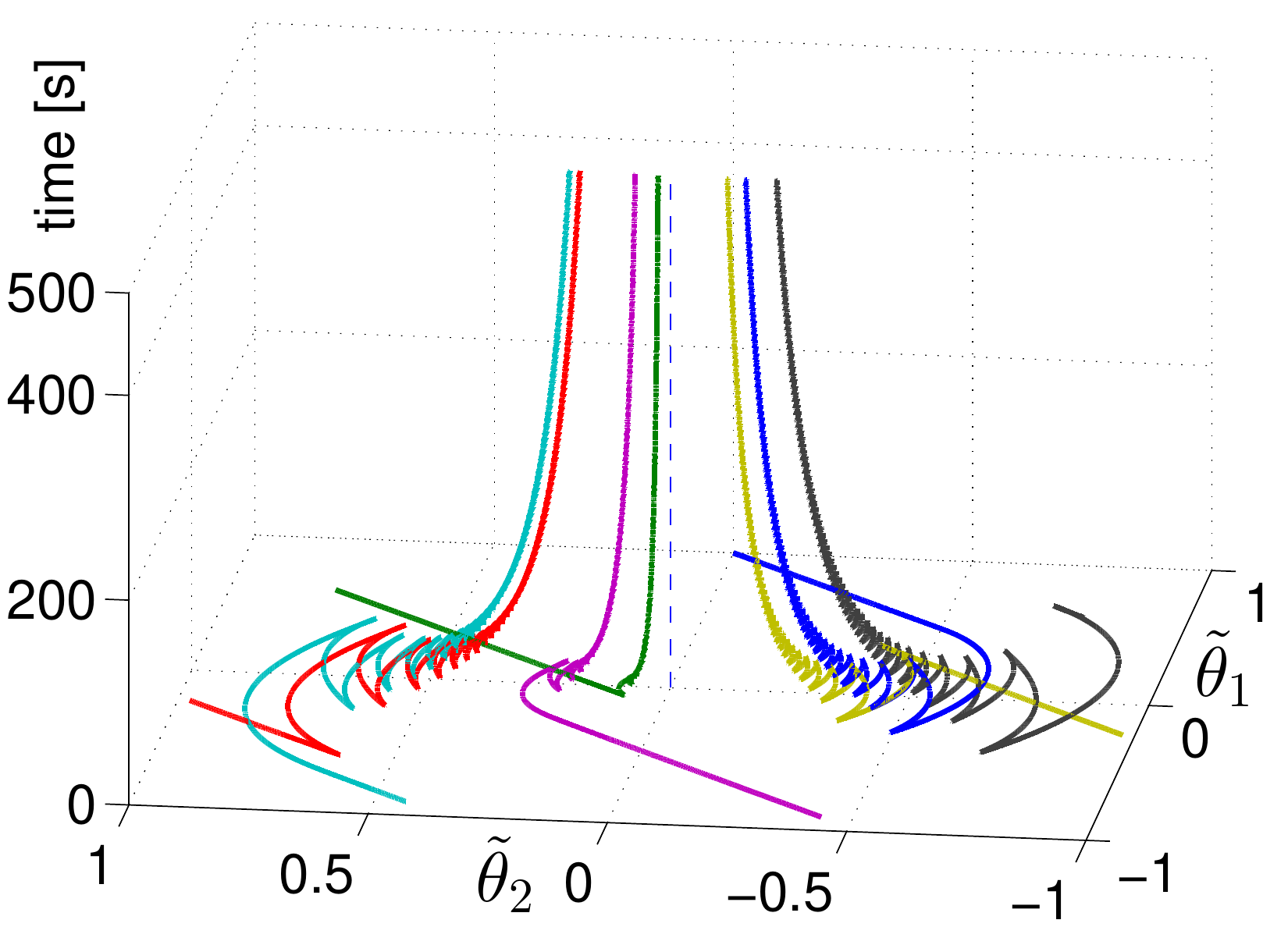}
	\caption{The integral curves of the parameter estimator \eqref{parest0} with $m(t)$ given by \eqref{eq:reg_m}, $\Gamma=3I_2$ and initial conditions taken in a disk.}
	\label{fig:classic_IC}
\end{figure}

Next we study performance of the DREM estimator \eqref{decest} with the same $m(t)$ and $\theta=\col(-3,3)$. The transient behavior of  $\tilde \theta(t)$ is given in Fig. \ref{fig:drgm_transients} for $\tilde \theta(0) = \col(3,-3)$, $\gamma_{1,2}=3$ and $\gamma_{1,2}=10$. The integral curves, for initial conditions taken in a disk, for $\gamma_{1,2}=3$ are given in Fig. \ref{fig:drgm_IC}. The simulations illustrate significant performance improvement both in oscillatory behavior and in convergence speed---notice the difference in time scales. Moreover, tuning of the gains $\gamma_i$ in the DREM estimator is straightforward. For example, given $\int_0^{10}{\phi^2(t)dt} \approx 0.78$, where $\phi$ is defined in \eqref{phiexa}, for $\gamma_i=3$ one obtains 
$$
	\tilde{\theta}_i(10) \approx e^{-3\cdot 0.78}\tilde{\theta}_i(0) \approx 0.09 \, \tilde{\theta}_i(0),
$$
and for $\gamma_i=10$
$$
	\tilde{\theta}_i(10) \approx e^{-10\cdot 0.78}\tilde{\theta}_i(0) \approx 0.0004 \, \tilde{\theta}_i(0),
$$
which coincides with the behaviour observed in Figs. \ref{fig:drgm_transients} (a) and (b). 

\begin{figure}[Htb]
	\centering
	\subcaptionbox{$\gamma_{1,2}=3$}{\includegraphics[width=0.45\textwidth]{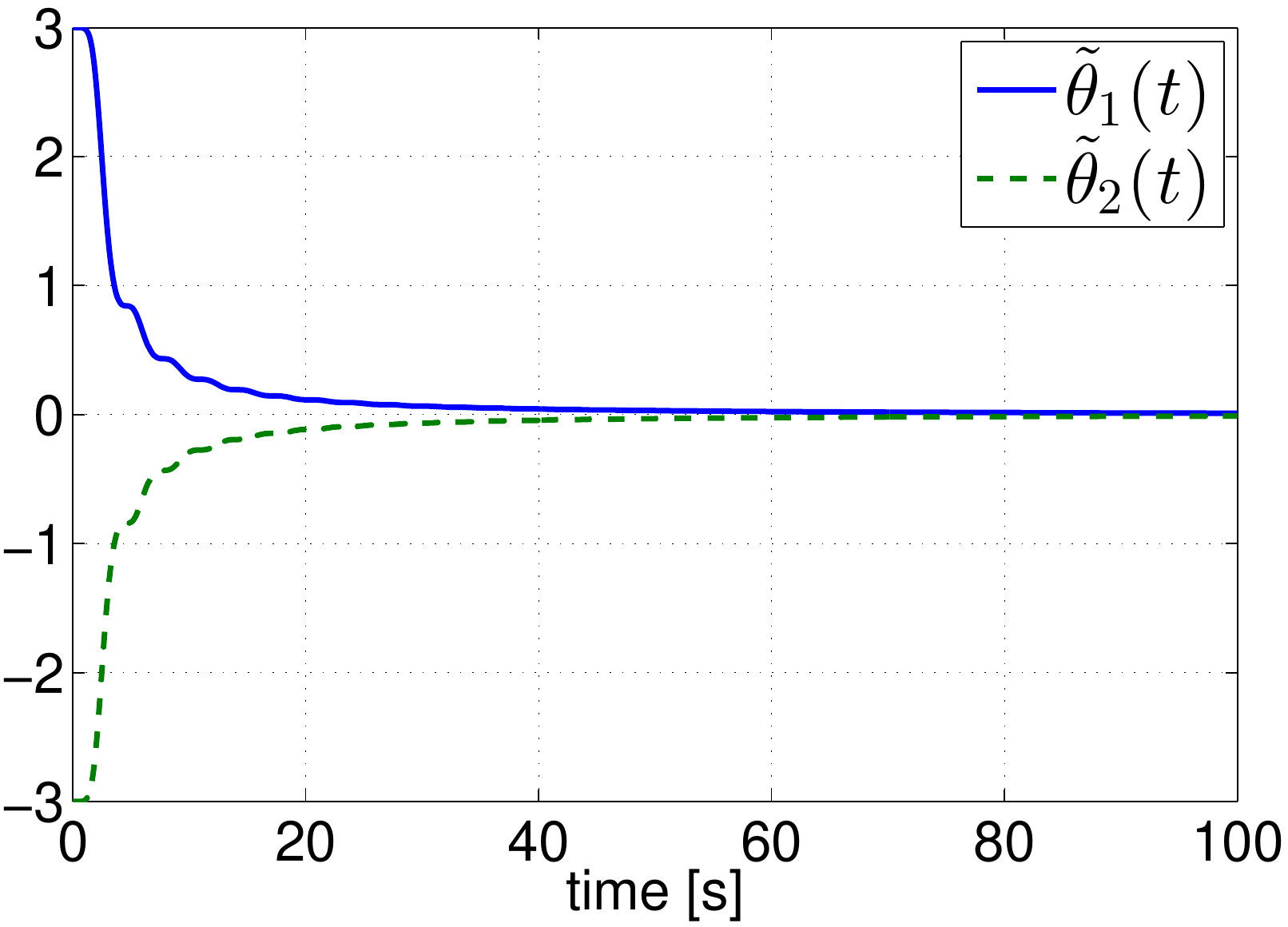}}
	\subcaptionbox{$\gamma_{1,2}=10$}{\includegraphics[width=0.45\textwidth]{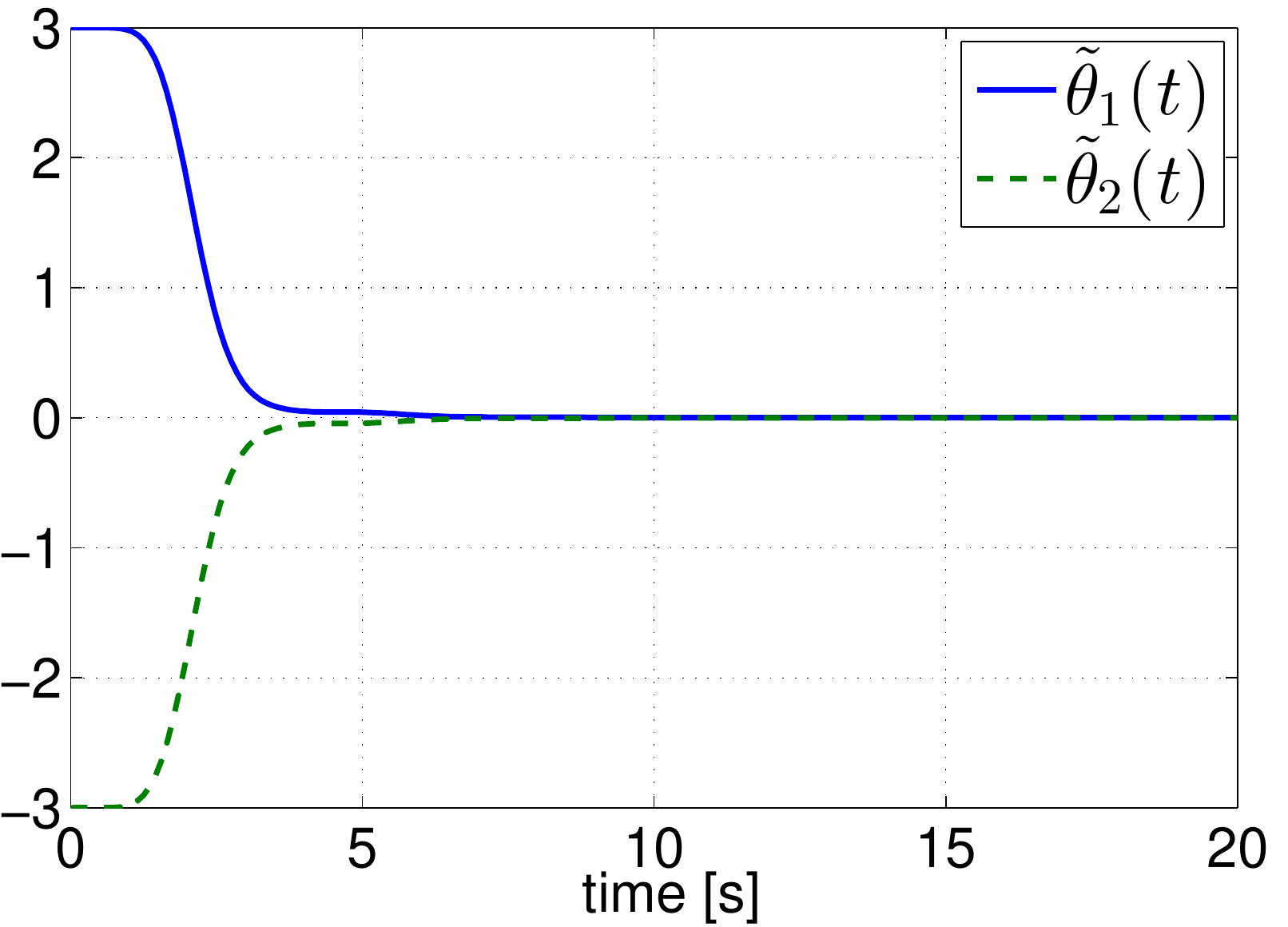}}
	\caption{Transient performance of the parameter errors $\tilde \theta(t)$ for the DREM estimator \eqref{decest} with $m(t)$ given by \eqref{eq:reg_m}, $\tilde \theta(0) = \col(3, -3)$ and $\gamma_{1,2}=3,10$. }
	\label{fig:drgm_transients}
\end{figure}

\begin{figure}[Htb]
	\centering
	\includegraphics[width=0.45\textwidth]{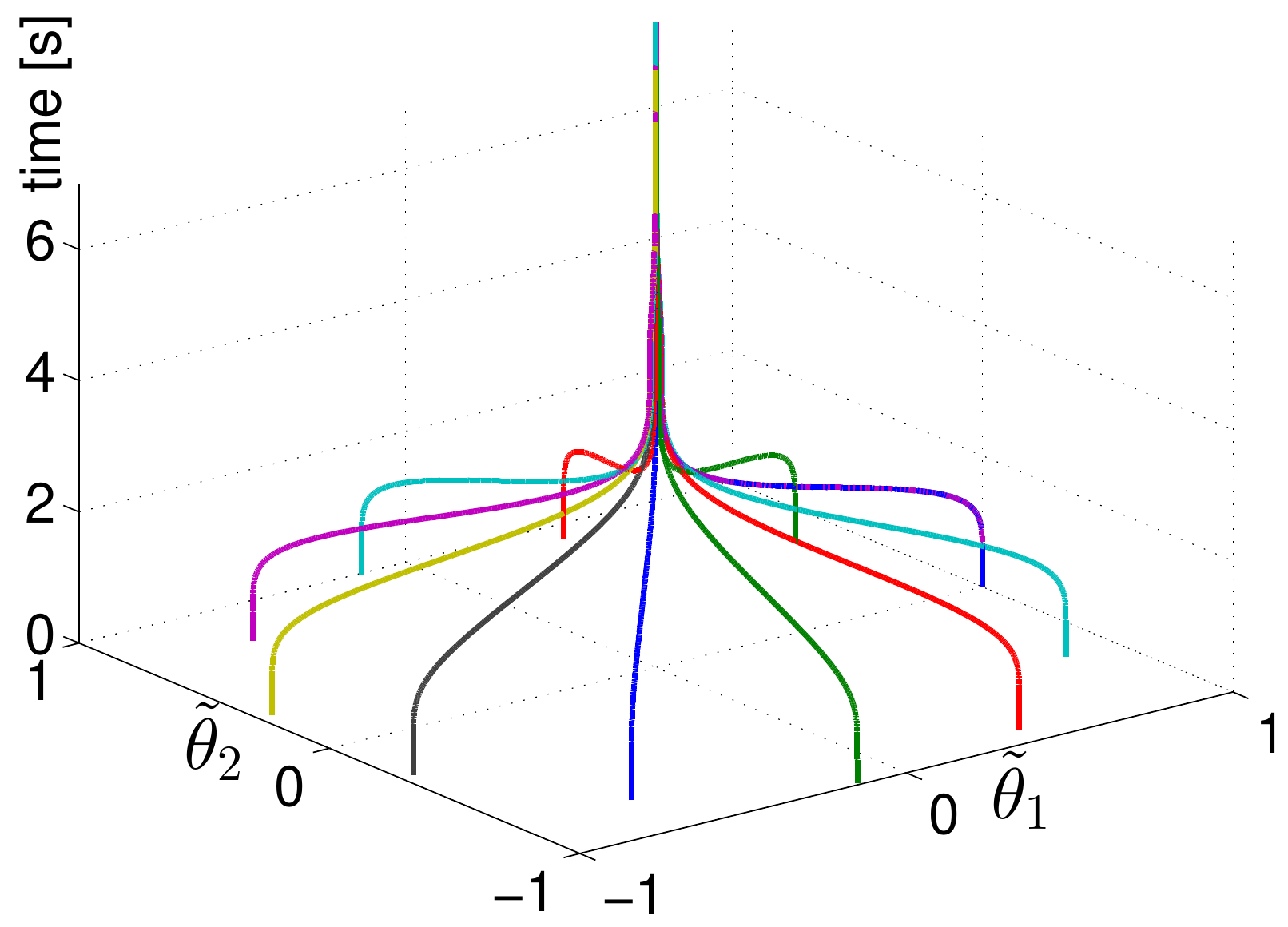}
	\caption{The integral curves of the DREM parameter estimator \eqref{decest} with $m(t)$ given by \eqref{eq:reg_m}, $\gamma_{1,2}=3$ and initial conditions taken in a disk.}
	\label{fig:drgm_IC}
\end{figure}

%%%%%%%%%%%%%%%%%%%%%
\section{Parameter Estimation of ``Partially" Monotonic Regressions}
\lab{sec3}
%%%%%%%%%%%%%%%%%%%%%%%%%%%%
%
In this section we propose to use the DREM technique for {\em nonlinearly} parameterised regressions with {factorisable} nonlinearities, where {\em some}---but not all---of the functions verify a monotonicity condition. The main objective is to generate a new regressor that contains only these ``good' nonlinearities. 

We consider factorisable regressions of the form
\begequ
\lab{facmod}
\bfy(t) = \bfm(t) \psi(\theta),
\endequ
where $\bfy:\rea_+ \to \rea^{n}$ and  $\bfm:\rea_+ \to \rea^{n \times p}$ contain measurable functions, the mapping $\psi:\rea^q \to  \rea^p$ is known and $\theta \in \rea^q$ is the {\em unknown} parameter vector.

For the sake of ease of exposition before presenting a general result, which is postponed to  Subsection \ref{subsec33}, in the next two subsections we apply the DREM technique to two representative examples. It is clear that the nonlinear regression \eqref{facmod} can be ``transformed" into a linear one  defining the vector $\eta:=\psi(\theta)$ to which the standard gradient estimator
\begequ
\lab{dothateta}
\dot {\hat \eta}=\Gamma \bfm^\top (\bfy - \bfm \hat \eta)
\endequ
can be  applied. In Subsection \ref{subsec34} a brief discussion on the disadvantages of overparameterisation is presented. To extend the realm of application of the DREM technique in Subsection \ref{subsec35} we consider the case of nonlinear dynamical systems that depend nonlinearly on unknown parameters---via non--factorisable nonlinearities. A procedure to generate  a factorisable regression of the form \eqref{facmod}, which is valid in a neighbourhood of the systems operating point, is given.  
\subsection{First example}
\lab{subsec31}
%%%%%%%%%%%%%%%%%%%%%%%%%%%%
%
To illustrate the use of DREM for nonlinearly parameterised factorisable regressions let us first consider the simplest scalar case of $n=1$, $p=2$ and $q=1$. The regression \eqref{facmod} becomes
\begequ
\lab{simreg}
	\bfy(t) = \begin{bmatrix}\bfm_1(t) & \bfm_2(t)\end{bmatrix}  \begin{bmatrix}\psi_1(\theta) \\ \psi_2(\theta) \end{bmatrix},
\endequ
where $\bfy:\rea_+ \to \rea$,  $\bfm_i:\rea_+ \to \rea$ and  $\psi_i:\rea \to \rea$, for $i=1,2$. Assume that $\psi_1(\theta)$ is {\em strongly monotonically increasing}, that is,
$$
\psi'_1(\theta) \geq \rho_0 > 0.
$$ 
In this case, the function $\psi_1$ verifies \cite{DEM} 
\begequ
\lab{moncon}
(a-b)[\psi_1(a)-\psi_1(b)] \geq \rho_1(a-b)^2,\quad \forall a,b \in \rea,
\endequ
for some $\rho_1>0$. 

The goal is to generate a new regression where, similarly to \cite{LIUTAC,LIUSCL,TYUetal}, the property \eqref{moncon} can be exploited. Following the DREM procedure described in Subsection \ref{subsec22} we apply an operator $H$ to \eqref{simreg} and pile--up the two regressions as 
\[
	\begin{bmatrix} \bfy(t) \\ \bfy_f(t) \end{bmatrix} = \begin{bmatrix} \bfm_1(t) & \bfm_2(t) \\ \bfm_{1f}(t) & \bfm_{2f}(t)\end{bmatrix} \begin{bmatrix}\psi_1(\theta) \\ \psi_2(\theta) \end{bmatrix}.
\]
Multiplying on the left the equation above by the row vector $[\bfm_{2f} \ -\bfm_2]$ we get the desired regression involving only $\psi_1$, namely,
$$
\bfY(t)=\Phi(t)\psi_1(\theta),
$$
where we defined the signals
\begequarr
\nonumber
\bfY & := & \bfm_{2f} \bfy - \bfm_2 \bfy_f\\
\lab{yphi}
\Phi & := & \bfm_{2f} \bfm_1 - \bfm_2 \bfm_{1f}.
\endequarr
The estimator
\begequ
\lab{estthe0}
	\dot{\hat{\theta}} = \gamma \Phi [\bfY- \Phi \psi_1(\hat\theta)],
\endequ
with $\gamma>0$, yields
\[
	\dot{\tilde{\theta}} = -\gamma \Phi^2 [\psi_1(\hat\theta) - \psi_1(\theta)].
\]
To analyse the stability of this error equation consider the Lyapunov function candidate
$$
V(\tilde \theta) = \frac{1}{2 \gamma} \tilde \theta^2,
$$
whose derivative yields
\begequarrs
\dot V  & = & - \Phi^2( \hat \theta - \theta)[ \psi_1(\hat\theta) - \psi_1(\theta)] \\
&  \leq &  - \rho_1 \Phi^2 (\hat \theta-\theta)^2 \\
&  = &  - 2\rho_1 \gamma\Phi^2 V,
\endequarrs
where the first inequality follows from \eqref{moncon}. Integrating the previous inequality yields
$$
V(t) \leq e^{- 2\rho_1\gamma \int_0^t \Phi^2(s)ds}V(0),
$$ 
which ensures that $ \tilde \theta(t) \to 0$ as $t \to \infty$ if $\Phi(t) \notin \call_2$.

As an example consider the regression 
\begequarrs
\bfy(t) & = & \bfm_1(t)\left(\theta - e^{-\theta} \right) + \bfm_2(t) \cos(\theta)\\
& = & \lef[{ccc} \bfm_1(t) & \; &  \bfm_2(t) \rig] \lef[{c} \theta - e^{-\theta} \\ \cos(\theta) \rig]\\
& =: & \bfm(t) \psi(\theta),
\endequarrs
which clearly satisfies condition \eqref{moncon}. The vector $\bfm(t) \notin \mbox{PE}$ stymies the application of overparameterisation. Moreover, since the mapping $\psi(\theta)$ is only locally injective, a constrained estimator of $\eta$ is required to recover the parameter $\theta$.
 
The proposition below identifies a class of regressors  $\bfm(t) \notin \mbox{PE}$ but $\Phi(t) \not \in \mathcal{L}_2$ for a simple delay operator.

%
%%%%%%%%%%%%%%
\begin{proposition}\em
\lab{pro3} 
The regressor
$$
\bfm(t) =\lef[{ccc} \frac{\sin(t)}{\sqrt{t+2\pi}} & \ & 1 \rig] \notin \mbox{PE}.
$$
Let the operator $H$ be the delay operator, that is,
$$
 (\cdot)_f (t)= (\cdot)(t-d),\quad d \in \left[\frac{\pi}{2}, \ \frac{3\pi}{2}\right].
$$ 
The function $\Phi$ defined in \eqref{yphi} verifies $\Phi(t) \not \in \mathcal{L}_2$. 
\end{proposition}
\begin{proof}
The fact that $\bfm(t)  \notin \mbox{PE}$ is obvious because $\bfm_1(t) \to 0$.

Now, the function $\Phi$ defined in \eqref{yphi} takes the form
\begequarrs
\Phi(t) & = &\frac{\sin(t)}{\sqrt{t+2\pi}} - \frac{\sin(t -d)}{\sqrt{t+2\pi - d}}.
\endequarrs
Whence
\[
	\begin{aligned}
		\Phi^2(t) &= \left( \frac{\sin(t)}{\sqrt{t+2\pi}} - \frac{\sin(t -d)}{\sqrt{t+2\pi - d}} \right)^2 \\
							&= \frac{\sin^2(t)}{{t+2\pi}} +  \frac{\sin^2(t -d)}{t+2\pi - d} 
								- 2 \frac {\sin(t)\sin(t-d)} {\sqrt{t+2\pi}\sqrt{t+2\pi - d}}\\
								& = \frac{\sin^2(t)}{{t+2\pi}} +  \frac{\sin^2(t -d)}{t+2\pi - d} \\
								& -2\cos(d)\frac{\sin^2(t)}{\sqrt{t+2\pi}\sqrt{t+2\pi - d}} \\
								& + \sin(d)\frac{\sin(2t)}{\sqrt{t+2\pi}\sqrt{t+2\pi - d}},
	\end{aligned}
\]
where some basic trigonometric identities have been used to derive the third identity. Note that the first three right hand terms of the last identity are not integrable. Since $\cos(d)\le 0$ in the admissible range of $d$ the sum of these terms is also not integrable. On the other hand, the last right term verifies  
$$
 \sin(d)\int_{0}^\infty{\frac{\sin(2t)}{\sqrt{t+2\pi}\sqrt{t+2\pi - d}}dt}< \infty.
$$
Thus, $\Phi(t) \not \in \mathcal{L}_2$.
\end{proof}

Simulations of the overparametrized estimator \eqref{dothateta} with $\theta=1$ are given in Fig. \ref{fig:nlp2_classic_transients}. The simulations exhibit poor convergence, which is expected since $m(t) \not \in \mbox{PE}$. After some trial--and--error gains tuning it is possible to improve the transient performance as shown in Fig. \ref{fig:nlp2_classic_transients}(b). However, as pointed out above, it is not possible to reconstruct $\theta$ since the function $\cos(\theta)$ is not injective for $\theta \in \mathbb{R}$.

\begin{figure}[Htb]
	\centering
\subcaptionbox{$\Gamma=\begin{bmatrix}3 & 0 \\ 0 & 3 \end{bmatrix}$}{\includegraphics[width=0.45\textwidth]{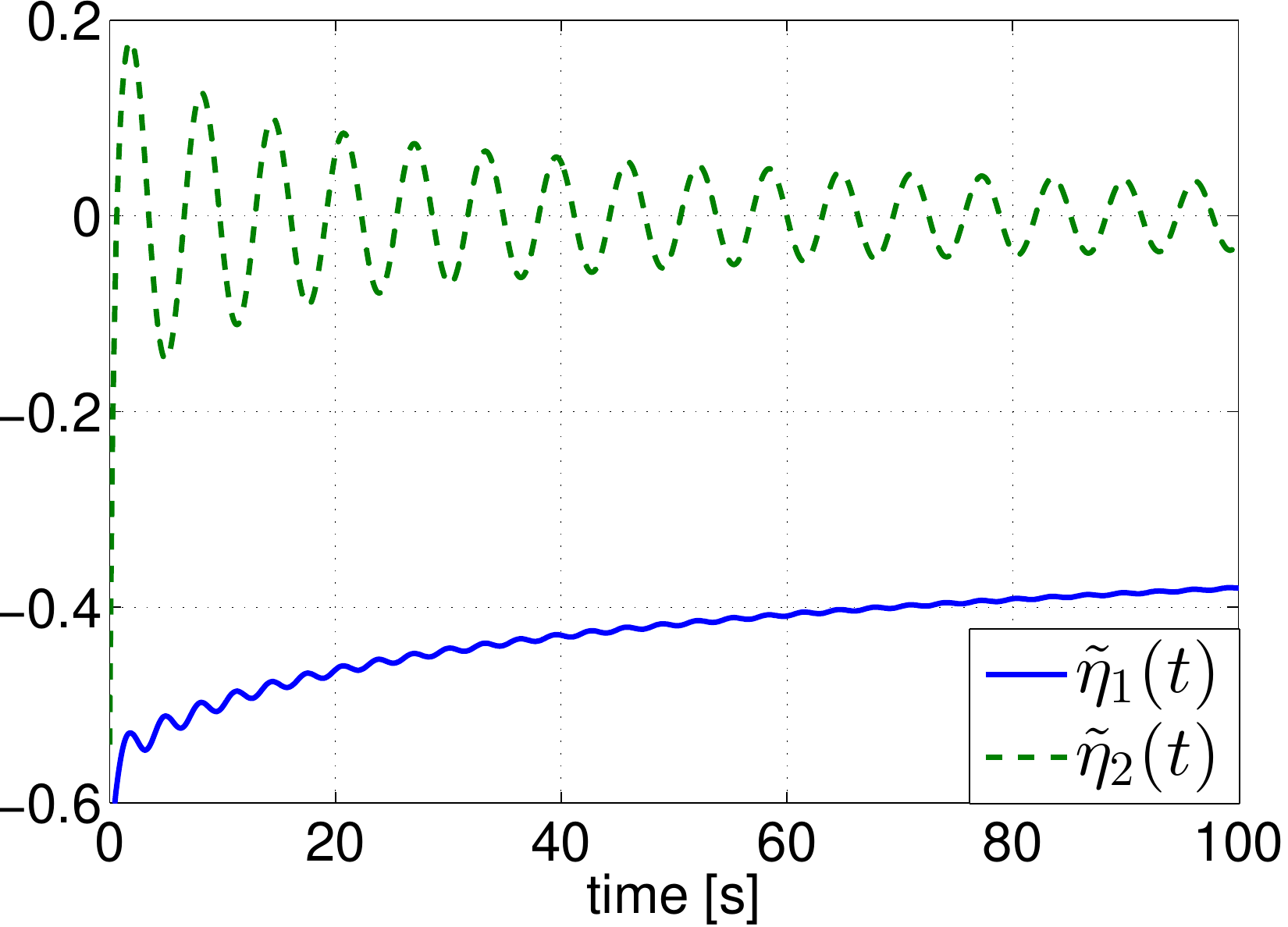}}
\subcaptionbox{$\Gamma=\begin{bmatrix}50 & 0 \\ 0 & 5 \end{bmatrix}$}{\includegraphics[width=0.45\textwidth]{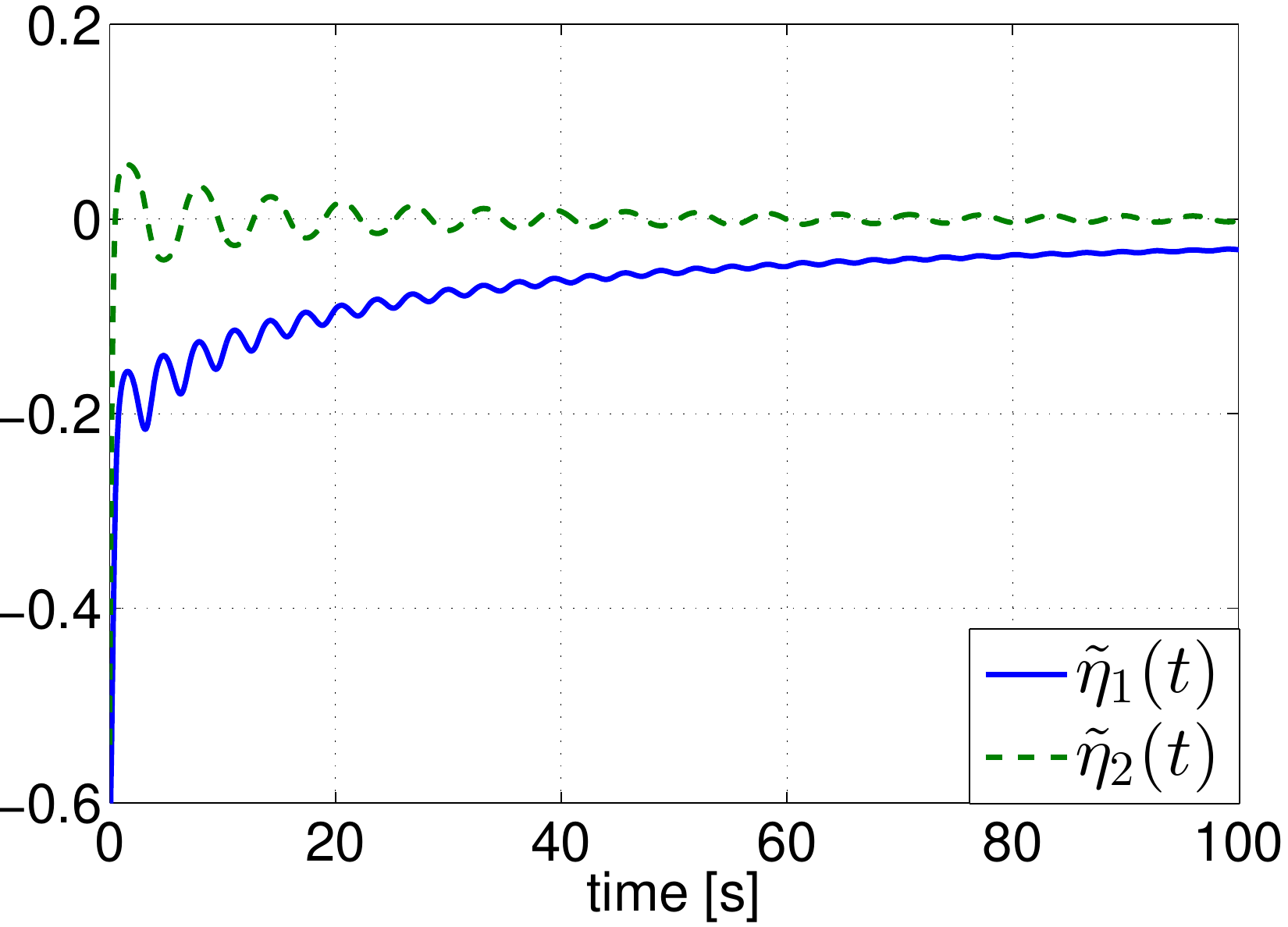}}
	\caption{Transient behaviour of the errors $\tilde \eta(t)$ for the over--parameterized parameter estimator \eqref{dothateta} for different adaptation gains; $\hat\eta(0)=0$ and $\tilde \eta(0) = -\psi(\theta) = -\psi(1) \approx [-0.63, \ -0.54]^\top$}.
	\label{fig:nlp2_classic_transients}
\end{figure}

Simulations of the DREM estimator \eqref{estthe0} with $\theta=1$ are shown in Fig. \ref{fig:nlp2_drem}. The estimation error $\tilde \theta(t)$ converges to zero and the tuning gain $\gamma$ allows to accelerate the convergence.

\begin{figure}[Htb]
	\centering
	\includegraphics[width=0.45\textwidth]{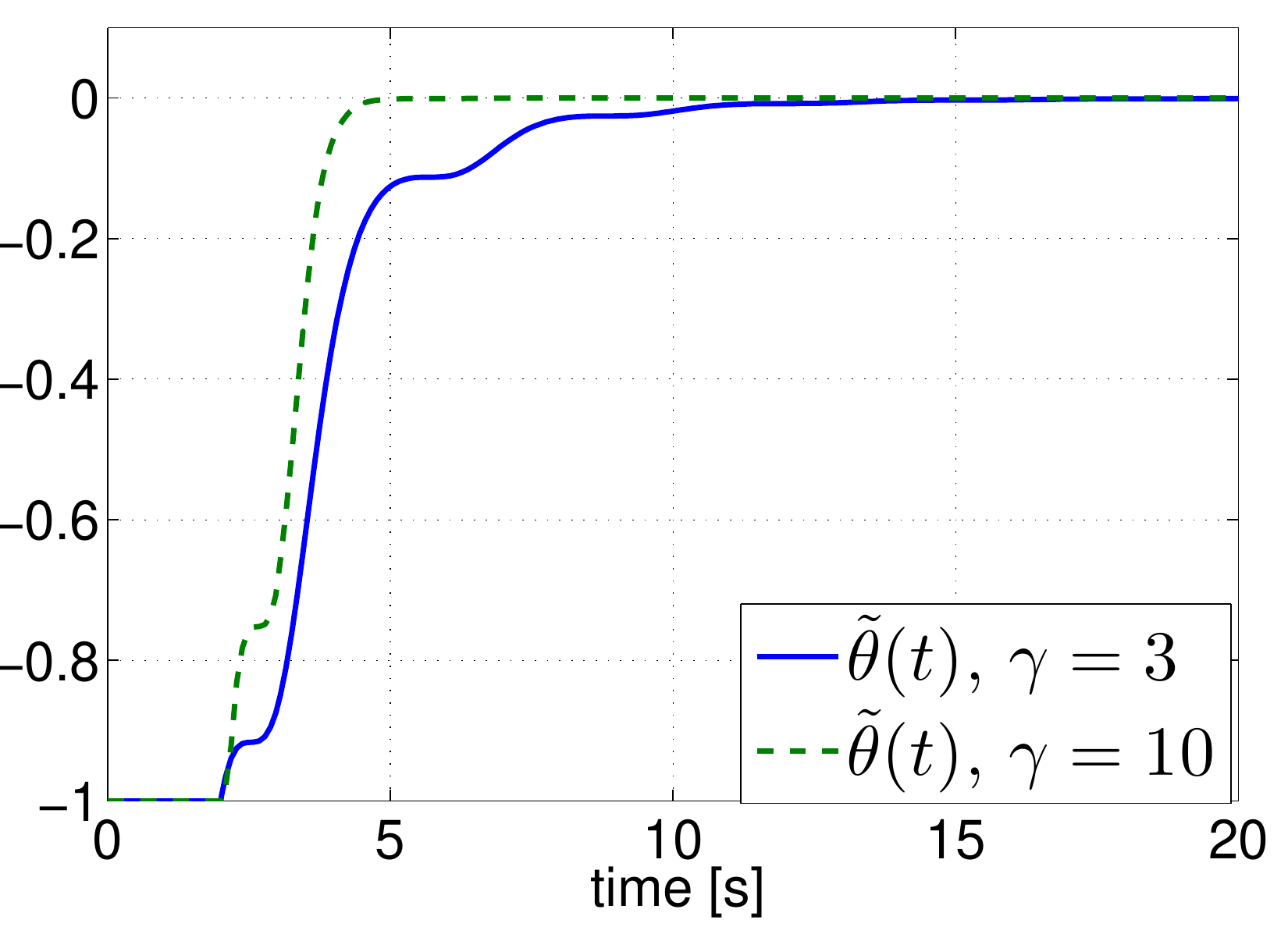}
	\caption{Transient behaviour of the error $\tilde \theta(t)$ transients of the DREM parameters estimator \eqref{estthe0} with different gains; $\tilde \theta(0)=-1$.}
	\label{fig:nlp2_drem}
\end{figure}

\begrem
If the strong monotonicity condition \eqref{moncon} is relaxed to {\em strict} (or plain)  monotonicity 
we can only ensure that $\dot V = - \Phi^2( \hat \theta - \theta)[ \psi_1(\hat\theta) - \psi_1(\theta)]\le 0$
and strict (PE--like) conditions must be imposed on $\Phi(t)$ to ensure convergence. 
\endrem

\begrem
Obviously, DREM is also applicable if $\psi_1$ is monotonically decreasing, instead of increasing, or if it is $\psi_2$ that enjoys the monotonicity property. 
\endrem

\subsection{Second example}
\lab{subsec32}
%%%%%%%%%%%%%%%%%%%%%%%%%%%%
%
Let us consider now the vector case with $n=2$, $p=3$ and $q=2$. The regression \eqref{facmod} becomes
\begequ
\lab{facmod1}
	\begin{bmatrix} \bfy_1(t) \\ \bfy_2(t) \end{bmatrix} = \begin{bmatrix} \bfm_{11}(t) & \bfm_{12}(t) & \bfm_{13}(t) \\ \bfm_{21}(t) & \bfm_{22}(t) & \bfm_{23}(t) \end{bmatrix} \begin{bmatrix}\psi_1(\theta) \\ \psi_2(\theta) \\\psi_3(\theta) \end{bmatrix}.
\endequ
where $\bfy_i:\rea_+ \to \rea$,  $\bfm_{ij}:\rea_+ \to \rea$ and  $\psi_j:\rea^2 \to \rea$, for $i=1,2$ and $j=1,2,3$. To estimate two parameters we assume that two functions $\psi_j$ satisfy a monotonicity condition---without loss of generality, we assume they are $\psi_2$ and $\psi_3$, and we define the ``good" vector 
$$
\psi_g:=\col(\psi_2,\psi_3).
$$
More precisely, we assume there exists a constant positive definite matrix $P \in \rea^{2 \times 2}$ such that 
\begequ
\lab{demcon}
P \nabla \psi_g(\theta) +  [\nabla \psi_g(\theta)]^\top P \geq  \rho_0 I_2>0.
\endequ

As shown in \cite{DEM}, see also \cite{LIUTAC}, the inequality above implies the existence of $\rho_1>0$ such that 
\begequ
\lab{moncon1}
	\left(a-b \right)^\top P \left[ \psi_g(a) - \psi_g(b) \right] \geq  \rho_1 |a - b|^2,\; \forall a,b\in \rea^q,
\endequ
which is the strong $P$--monotonicity property exploited in  \cite{LIUSCL}. 

The first task of DREM is to generate a two--dimensional regression that does not contain $\psi_1$. First we apply a filter to either one of the rows of \eqref{facmod1}, say to the first one, and pile--up the result as
\[
\begin{bmatrix} \bfy_1(t) \\ \bfy_2(t) \\ \bfy_{1f}(t) \end{bmatrix} = \begin{bmatrix} 
			\bfm_{11}(t) & \bfm_{12}(t) & \bfm_{13}(t) \\ 
			\bfm_{21}(t) & \bfm_{22}(t) & \bfm_{23}(t) \\
			\bfm_{11f}(t) & \bfm_{12f}(t) & \bfm_{13f}(t) \end{bmatrix} \begin{bmatrix}\psi_1(\theta) \\ \psi_2(\theta) \\\psi_3(\theta) \end{bmatrix}.
\]
Next we multiply the last equation on the left by the matrix
\[
	\begin{bmatrix} \bfm_{21}(t) & -\bfm_{11}(t) & 0 \\ \bfm_{11f}(t) & 0 & -\bfm_{11}(t)\end{bmatrix},
\]
which is the left annihilator of the first column of the extended regressor. Defining the $2 \times 2$ regression matrix
$$
\Phi:= \begin{bmatrix} \bfm_{21} & -\bfm_{11} & 0 \\ \bfm_{11f} & 0 & -\bfm_{11}\end{bmatrix}
	\begin{bmatrix} 
			\bfm_{12} & \bfm_{13} \\ 
			\bfm_{22} & \bfm_{23} \\
			\bfm_{12f} & \bfm_{13f} \end{bmatrix},
$$
and 
$$
\bfY_1:=\begin{bmatrix} \bfm_{21} & -\bfm_{11} & 0 \\ \bfm_{11f} & 0 & -\bfm_{11}\end{bmatrix}
									\begin{bmatrix}\bfy_1 \\ \bfy_2 \\ \bfy_{1f}\end{bmatrix},
$$ 
we get  
\begequ \lab{eq:Y1}
	\bfY_1(t)= \Phi(t) \psi_g(\theta).
\endequ
Multiplying \eqref{eq:Y1} on the left by $\adj\{\Phi(t)\}$ and defining $\bfY(t):=\adj\{\Phi(t)\}\bfY_1(t)$, we obtain the desired regression form
\begequ
\lab{regfor2}
\bfY(t)= \det\{\Phi(t)\} \psi_g(\theta).
\endequ

We propose the estimator
\begequ
\lab{estthe}
	\dot{\hat{\theta}} = \det\{\Phi\} \Gamma P [\bfY- \det\{\Phi\} \psi_g(\hat\theta)],
\endequ
with $\Gamma \in \rea^{2 \times 2}$ a positive definite gain matrix. Using \eqref{regfor2} the error equation is
\[
	\dot{\tilde{\theta}} = -{\det}^2\{\Phi\} \Gamma P [\psi_g(\hat\theta) - \psi_g(\theta)].
\]
To analyse its stability define the Lyapunov function candidate
\begequ
\lab{lyafun}
V(\tilde \theta) = \frac{1}{2} \tilde \theta^\top \Gamma^{-1} \tilde \theta,
\endequ
whose derivative yields
\begequarrs
\dot V  & = & - {\det}^2\{\Phi\} ( \hat \theta - \theta)^\top P [ \psi_g(\hat\theta) - \psi_g(\theta)] \\
& \leq & - {\det}^2\{\Phi\} \frac{2 \rho_1}{\lambda_{\max}\{\Gamma\}}V.
\endequarrs
If the matrix $\Phi(t)$ is full rank and ${\det}^2\{\Phi(t)\}\ge \kappa>0$, then
\[
	\dot V  \leq - \frac{2 \kappa \rho_1}{\lambda_{\max}\{\Gamma\}}V,
\]
and the analysis above ensures {\em exponential} stability of the error equation.

Otherwise integrating the inequality yields 
$$
V(t) \leq e^{-  \frac{2 \rho_1}{\lambda_{\max}\{\Gamma\}} \int_0^t \det^2\{\Phi(s)\}ds}V(0),
$$ 
which ensures that $ \tilde \theta(t) \to 0$ as $t \to \infty$ if  $\det\{\Phi(t)\} \notin \call_2$.

\begrem
From the derivations above it is clear that there are many degrees of freedom for the definition of the regressor matrix $\Phi$, {\em e.g.}, the choice of the operator $H$, the rows to be filtered and mixed. The final objective is to impose some properties to the function $\det\{\Phi(t)\}$, ideally, that it is (uniformly) bounded away from zero---that yields exponential stability.
\endrem

\begrem
In \cite{LIUSCL} it is shown that the {\em local} verification of the monotonicity condition \eqref{demcon} reduces to a linear matrix inequality (LMI) test provided some prior knowledge on the parameters is available. More precisely, assume $\theta \in \Theta \subset \rea^q$, with
$$
\Theta :=\{ \theta \in \rea^q\;|\; \theta_i \in [\theta_i^m,\theta_i^M] \subset \rea\}.
$$
The {\em quadratic approximation} of the  mapping $\psi_g(\theta)$ verifies (\ref{demcon}) if and only if the LMI
$$
P \nabla \psi_g(v_i) + [\nabla \psi_g(v_i)]^\top P >  0,\quad i =1,\dots,(2^q)^q
$$
is feasible, where the vectors $v_i \in \Theta$ are computable from the vertices of $\Theta$. 
\endrem
\subsection{A general result}
\lab{subsec33}
%%%%%%%%%%%%%%%%%%%%%%%%%%%%
%
In this subsection a generalization of the previous examples is presented. We make the following assumption.

\begin{assumption}\em
\lab{ass1}
Consider the regression form \eqref{facmod}. There are $q$ functions
$\psi_i$ that, reordering the outputs $y_i$, we arrange in a vector $\psi_g:\rea^q \to \rea^q$, verifying  
$$
P \nabla \psi_g(\theta) +  [\nabla \psi_g(\theta)]^\top P \geq  \rho_0 I_q>0,
$$ 
for some  positive definite matrix $P \in \rea^{q \times q}$.
\end{assumption}   

Consistent with Assumption \ref{ass1} we rewrite  \eqref{facmod} as 
\begequ
\lab{yn}
\bfy_N(t)  = \begin{bmatrix} \bfm_{g}(t) & \bfm_{b}(t) \end{bmatrix} \begin{bmatrix}\psi_g(\theta) \\ \psi_b(\theta) \end{bmatrix},
\endequ
where $\bfy_N:\rea_+ \to \rea^n$ is the reordered output vector,  $\bfm_{g}:\rea_+ \to \rea^{n \times q}$,   $\bfm_{b}:\rea_+ \to \rea^{n \times (p-q)}$, $\psi_g:\rea^q \to \rea^q$ and  $\psi_b:\rea^q \to \rea^{p-q}$. 

From the previous two examples we have learned that DREM must accomplish two tasks, on one hand, generate a regression without $\bfm_b$. On the other hand, to be able to relax the PE condition, the new regressor matrix should be square (or tall). Given these tasks, to obtain a sensible problem formulation the following assumption is imposed.

\begin{assumption} \em
\lab{ass2}
The regression \eqref{yn} satisfies
\begequarr
\lab{qlesp}
q & < & p \\
\lab{nlesp}
n & < & p.
\endequarr
\end{assumption}

If \eqref{qlesp} does not hold {\em all} functions $\psi_i,\;i=1,\dots,p$, satisfy the monotonicity condition and there is no need to eliminate any one of them. On the other hand, if \eqref{nlesp} is not satisfied a square regressor without the ``bad" part of the regressor $\psi_b$ can be created without the introduction of the operators $H_i$. Indeed, if $n = p$ the matrix $\bfm_b$ is tall and it admits a full--rank left annihilator $\bfm_b^\perp:\rea_+ \to \rea^{q \times n}$. Moreover, the new regressor matrix $\bfm_b^\perp \bfm_g$ is square. A similar situation arises if $n > p$.   

Following DREM we introduce $n_f$ operators, apply them to some rows of \eqref{yn} and pile all the regression forms to get
\begin{equation} \label{eq:ExtReg}
	\lef[{c} \bfy_N \\ \bfy_{Nf} \rig]  =  \begin{bmatrix} \bfM_{g} & \bfM_{b} \end{bmatrix} \begin{bmatrix}\psi_g(\theta) \\ \psi_b(\theta) \end{bmatrix}.
\end{equation}
where we defined the matrices  $\bfM_{g}:\rea_+ \to \rea^{(n+n_f) \times q}$,   $\bfM_{b}:\rea_+ \to \rea^{(n+n_f) \times (p-q)}$
\begequ
\lab{bfmat} 
\bfM_{g}:=  \begin{bmatrix} \bfm_{g} \\  \bfm_{gf}  \end{bmatrix},\; \bfM_{b}:=\begin{bmatrix}  \bfm_{b}\\  \bfm_{bf} \end{bmatrix}.
\endequ
To select the number $n_f$ of operators we notice that the matrix to be eliminated, that is $\bfM_b$,
is of dimension $(n+n_f) \times (p-q)$. Therefore, to have a left annihilator for it with $q$ rows, which is needed to make the new regressor square, we must fix $n_f=p-n$. Define
\begin{equation} \label{eq:Phi}
	\Phi := \bfM_b^\perp  \bfM_{g}.
\end{equation}
Multiplying on the left by $\adj\{\Phi\}\bfM_b^\perp$ the equation \eqref{eq:ExtReg} yields the desired regressor form
$$
\bfY = \det\{\Phi\} \psi_g(\theta),
$$
where
\begequ \lab{bfyphi}
\bfY := \adj\{\Phi\}\bfM_b^\perp \lef[{c} \bfy_N \\ \bfy_{Nf} \rig]
\endequ

We are in position to present the main result of this section, whose proof follows from the derivations above. 

\begin{proposition}\em
\lab{pro4}
Consider the nonlinearly parameterised factorisable regression \eqref{yn} satisfying Assumptions \ref{ass1} and \ref{ass2}. Introduce  $p-n$ linear, $\call_\infty$--stable operators $H_i: \call_\infty \to \call_\infty,\;i \in \{1,2,\dots,p-n\}$ verifying \eqref{defh}. Define the matrices $\bfM_g,\;\bfM_b$ as given in \eqref{bfmat}. Consider the estimator \eqref{estthe} with $\Phi$ and $\bfY$ defined in \eqref{eq:Phi}, \eqref{bfyphi} and $\bfM_b^\perp:\rea_+ \to \rea^{q \times p}$ a full--rank left annihilator of $\bfM_b$. The following implication  holds
$$
\det\{\Phi(t)\} \notin \call_2 \quad \Longrightarrow \quad \lim_{t \to \infty}| \tilde \theta(t)| = 0.
$$ 
Moreover, if ${\det}^2\{\Phi(t)\} \ge \kappa >0$, then $| \tilde \theta(t)|$ tends to $0$ {\em exponentially} fast.
\end{proposition}
\subsection{Discussion on overparameterisation}
\lab{subsec34}
%%%%%%%%%%%%%%%%%%%%%%%%%%%%
%
As disused above the nonlinear regression \eqref{facmod} can be ``transformed" into a linear one  defining the vector $\eta=\psi(\theta)$. Classical
parameter adaptation algorithms can then be used to estimate this new parameter vector.  However, overparametrization suffers from the following well--known shortcomings
\cite{LJU,SASBOD}:
\begite
\item[(i)] Performance degradation, {\em e.g.}, slower convergence, due to the need of a search in a bigger parameter space---indeed, usually $q < p$.
\item[(ii)] The more stringent conditions imposed on the reference signals to ensure the PE needed for convergence of the parameters $\eta$.
\item[(iii)] Inability to recover the true parameter $\theta$---except for injecting mappings. This stymies the application of this approach in situations where the actual parameters $\theta$ are needed.
\item[(iv)] Conservativeness introduced when incorporating prior knowledge in restricted parameter estimation.
\item[(v)] Reduction of the domain of validity of the estimates stemming from the, in general only local, invertibility of the
mappings $\psi$.
\endite

\subsection{Generating separable regressions for nonlinear dynamical systems}
\lab{subsec35}
%%%%%%%%%%%%%%%%%%%%%%%%%%%%
%
Consider the general case of a nonlinear system whose dynamics depend nonlinearly on some unknown parameters, say,
\begequ
\lab{dotx}
\dot x = F_0(x)+F_1(x,\theta),
\endequ
where $F_0:\rea^n  \to \rea^n$, $F_1:\rea^n \times \rea^q \to \rea^n$ are known functions, $x:\rea_+ \to \rea^n$  is the measurable systems state and $\theta \in \rea^q$ are the constant unknown parameters. It is assumed that the system evolves around some operating point $x_* \in \rea^n$. 

Applying the classical filtering technique \cite{MIDetal} it is possible to generate a  state--dependent regression that can be used for parameter estimation. Indeed, defining
\begequarrs
\bfy& := & \frac{p}{p+1}x - \frac{1}{p+1}F_0(x)\\
\Xi(x,\theta) & := & \frac{1}{p+1}F_1(x,\theta),
\endequarrs
and neglecting exponentially decaying terms, we get the nonlinearly parameterized regression 
\begequ
\lab{nonlinreg}
\bfy = \Xi(x,\theta)
\endequ
where  $\bfy:\rea_+ \to \rea^n$ and $\Xi:\rea^n \times \rea^q \to \rea^n$. Estimation of the parameters at this level of generality is a daunting task. To simplify it we propose to linearize the function $\Xi$ around $x_*$ to obtain a {\em separable} nonlinear regression as done in the lemma below.

\begin{lemma}\em
\lab{lem1}
The first order approximation (around $x_*$) of the nonlinearly parameterized regression \eqref{nonlinreg} is given by \eqref{facmod}, where $\bfm:\rea_+ \to \rea^{n \times p}$ and $\psi:\rea^q \to  \rea^p$, with $p:=n+n^2$, are given by
\begequarrs
\bfm(t) & := & \lef[{ccc} I_n & | & \lef[{cccc} \tilde x^\top(t) & 0 & \dots & 0 \\  0 & \tilde x^\top(t) &  \dots & 0 \\ \vdots & \vdots & \vdots & \vdots \\ 0 & 0 & \dots & \tilde x^\top(t)\rig]\rig] \\
\psi(\theta)  & := & \lef[{c} \Xi_*(\theta) \\  \nabla_x \Xi_{1*}(\theta) \\ \vdots \\ \nabla_x \Xi_{n*}(\theta)\rig],
\endequarrs
where we have introduced the notation $\tilde x:=x-x_*$.
\end{lemma}
\begin{proof}
The proof follows trivially rearranging the elements of the first order approximation (around $x_*$) of $\Xi$, namely
$$
\Xi(x,\theta)= \Xi(x_*,\theta)+\nabla^\top_x\Xi(x_*,\theta)(x - x_*) + \calo (|x - x_*|^2),
$$
and using the definition of $\tilde x$.
\end{proof}
Next, from the derivations of the subsection \ref{subsec33}, if at least $q$ elements of the vector $\psi(\theta)$ enjoy monotonicity, the DREM approach can be applied.
\begrem
From the structure of the regression matrix given in Lemma \ref{lem1} the following equivalence holds true
$$
\bfm(t) \in \mbox{PE} \quad \Longleftrightarrow \quad \lef[{c} 1 \\ \tilde x(t) \rig] \in \mbox{PE}.
$$
Clearly, the PE condition is {\em not satisfied} in regulation tasks, when $\tilde x(t)$ tends to zero.
\endrem

\begrem
We have generated the regressor form  \eqref{nonlinreg} applying the classical filtering technique. As shown in \cite{LIUTAC,LIUSCL} it is possible to design parameter estimators directly for the system \eqref{dotx} using immersion and invariance techniques \cite{ASTKARORT}. Also, we have considered the case of a closed system, that is, without input. The same derivations apply for systems with an external input signal.  
\endrem
%
%%%%%%%%%%%%%%%%%%%%%
\section{Concluding Remarks and Future Research}
\lab{sec4}
%%%%%%%%%%%%%%%%%%%%%%%%%%%%
%
A procedure to generate new regression forms for which we can design parameter estimators with enhanced performance has been proposed. The procedure has been applied to linear regressions yielding new estimators whose parameter convergence can be established without invoking the usual, hardly verifiable, PE condition. Instead, it is required that the new regressor vector is not square integrable, which is different than PE of the original regressor. For nonlinearly parameterised regressions with monotonic nonlinearities the procedure allows to treat cases when only some of the nonlinearities verify this monotonicity condition. Similarly to the case of linear regressions, convergence is ensured if the determinant of the new regressor is not square integrable.

The design procedure includes many degrees of freedom to verify the aforementioned convergence condition. Current research is under way to make more systematic the choice of this degrees of freedom. It seems difficult to achieve this end at the level of generality presented in the paper. Therefore, we are currently considering more ``structured" situations, for instance, when the original regression form comes from classes of physical dynamical systems or for a practical application. Preliminary calculations for the problem of current--voltage characteristic of photovoltaic cells---which depend nonlinearly on some unknown parameters---are encouraging and we hope to be able to report the results soon.  

\bibliographystyle{IEEEtran}
\bibliography{IEEEabrv,BibList}

\end{document}